\documentclass[english,11pt,letterpaper]{article}
\usepackage[margin=1.0in]{geometry}
\usepackage[numbers]{natbib}
\usepackage{algorithm}
\usepackage{algorithmic}

\usepackage{multirow}
\usepackage[pdftex,colorlinks]{hyperref}
\usepackage{float}
\usepackage{amssymb}
\usepackage{amsmath}
\usepackage{amsfonts}
\usepackage{amssymb}
\usepackage{amsthm}
\usepackage{algorithm}
\usepackage{algorithmic}
\usepackage{graphicx}

\newtheorem{theorem}{Theorem}[section]
\newtheorem{lemma}[theorem]{Lemma}

\newtheorem{counter-example}[theorem]{Counter example}

\newtheorem{open question}[theorem]{Open question}

\makeatletter

\floatstyle{ruled}
\newfloat{algorithm}{tbp}{loa}
\providecommand{\algorithmname}{Algorithm}
\floatname{algorithm}{\protect\algorithmname}

\makeatother

\usepackage{babel}

\begin{document}

\begin{titlepage}

\author{
	Michael Schapira\thanks{School of Computer Science and Engineering, The Hebrew University of Jerusalem, Israel.} \hspace{1cm} 
	Gal Shahaf\thanks{Dept. of Mathematics, The Hebrew University of Jerusalem, Israel}
}

\title{Oblivious Routing via Random Walks}
\date{}

\maketitle

\begin{abstract}	
We present novel oblivious routing algorithms for both splittable and unsplittable multicommodity flow. Our algorithm for minimizing congestion for \emph{unsplittable} multicommodity flow is the
first oblivious routing algorithm for this setting. As an intermediate step towards this algorithm, we present a novel generalization of Valiant's classical load balancing scheme for packet-switched networks to arbitrary graphs, which is of independent interest. Our algorithm for minimizing congestion for \emph{splittable} multicommodity flow improves upon the state-of-the-art, in terms of both running time and performance, for graphs that exhibit good expansion guarantees. Our algorithms rely on diffusing traffic via iterative applications of the random walk operator. Consequently, the performance guarantees of our algorithms are derived from the convergence of the random walk operator to the stationary distribution and are expressed in terms of the spectral gap of the graph (which dominates the mixing time).
\end{abstract}
\end{titlepage}

\section{Introduction}\label{sec: intro}

Oblivious routing is the task of routing traffic in a manner that is agnostic to the traffic demands. Due to its inherent robustness to uncertainty about traffic conditions, varying traffic patterns, and inaccurate traffic measurements, the design of oblivious routing algorithms is of great theoretical interest and practical importance~\cite{zhang2005designing,mckeown2009software,xu1994iterative,greenberg2008towards,greenberg2009vl2,applegate2003making}. Consequently, a rich body of algorithmic literature is focused on devising oblivious routing schemes with provably high performance guarantees. Surprisingly, as first established by Valiant~\cite{valiant1981universal,valiant1982scheme}, and later by R{\"a}cke~\cite{racke2002minimizing,racke2008optimal}, oblivious routing can provide good approximations to the performance of the optimal demands-aware routing.

We revisit this long-standing algorithmic challenge and present oblivious routing algorithms for two well-studied settings: (1) minimizing congestion for \emph{splittable} multicommodity flow (s-MCF), and (2) minimizing congestion for \emph{unsplittable} multicommodity flow (u-MCF). Our algorithm for u-MCF is the first oblivious routing scheme for this setting. As an intermediate step towards this algorithm, we present an novel algorithm for minimizing delay in Valiant's classical model for routing permutation traffic demands in packet-switched networks. This algorithm extends the classical Valiant Load Balancing scheme, designed for the hypercube, to arbitrary graphs. Our results for s-MCF improve upon today's state-of-the-art, in terms of both performance and running time, for graphs that exhibit good expansion guarantees. We next discuss our high-level approach and then delve into our results for each of the models.

\noindent{\bf High-level approach.} Our approach to oblivious routing reflects a simple high-level idea: utilizing the efficient distribution of traffic across the network by the random walk operator. Intuitively, our algorithms first disperse the traffic from the source uniformly across the network and then regather the traffic at the destination. The realization of this simple intuition, however, requires care, and is different across the three settings. Leveraging random walks for routing can be traced back to the work of Broder et al.~\cite{broder1994existence,broder1997static} on finding edge-disjoint paths and establishing virtual circuits and to studied of the minimization of congestion and load balancing in this context~\cite{akyildiz2000new,tian2005randomwalk,spyropoulos2005spray,servetto2002constrained}. In contrast to these studies, which compute routing outcomes for \emph{input} traffic demands, our focus is on traffic-\emph{oblivious} routing.

Under our algorithm for the s-MCF model, traffic dispersion from the source is obtained via iterative applications of the random walk operator. The key challenge lies in guaranteeing that all flow eventually reach the destination. To accomplish this, a series of stochastic operators is applied to ``invert'' each step of a random walk that starts at the destination. Our algorithms for the u-MCF model and Valiant's model generate a ``sample space'' of paths by sampling multiple fixed-length random walks from each vertex. Then, a \emph{single} path from each source $s$ to destination $t$ is computed by randomly selecting a vertex $v$ and concatenating two paths: (1) a randomly selected path from the sample space connecting $s$ to $v$, and (2) a randomly selected path from the sample space connecting $v$ to $t$. Careful analyses show that this routing scheme results in provably low congestion in both models.

Our analyses heavily rely on the convergence of the random walk to the stationary distribution. Consequently, the performance guarantees and the running times of our algorithms are expressed in terms of the spectral gap of the graph (which dominates the mixing time).

Before presenting our results, we first introduce the following notation: We consider undirected and capacitated graphs $G=(V,E,c)$ of size $n=|V|$. The \emph{degree} of each vertex $x\in V$ is defined to be $d_x:=\sum_{(x,y)\in E} c(x,y)$. Let $\pi(G)\in \mathbb{R}^n$ be the \emph{stationary distribution}: $\pi_x=\frac{d_x}{\sum_{y\in V} d_y}$. We denote by $d_{min}$ and $d_{max}$ the minimal and maximal degrees, respectively, and by $\pi_{min}$ and $\pi_{max}$ the minimal and maximal value of a vertex in $\pi$, respectively. We say that $G$ is \textit{$d$-regular} if $d_x=d$ for all $x\in V$, and $c(e)=1$ for all $e\in E$. Let $\lambda(G)$ be the generalized second eigenvalue of the random walk operator of $G$ (see Appendix \ref{app: RW prelim} for formal definitions), and let $\bar{\lambda}(G):=\min \{\lambda(G),\lambda(G')\}$, where $G'$ is the graph obtained from $G$ via the addition of self loops. When clear from the context we will refer to $\pi(G)$, and $\bar{\lambda}(G)$ simply as $\pi$ and $\bar{\lambda}$. 





\noindent{\bf Results for unsplittable multicommodity flow.} In the unsplittable multicommodity flow setting, the routing of each commodity must use a single path in the network~\cite{chakrabarti2007approximation,erlebach2006approximation,alvelos2003comparing,kolliopoulos1997improved}. We present the first (to the best of our knowledge) oblivious routing algorithm for this context. We prove the following upper bound on the \emph{performance ratio}~\cite{racke2002minimizing,racke2008optimal} of this algorithm, i.e., the maximum ratio between the congestion (max link utilization) induced by the algorithm and the congestion under the optimal \emph{demands-aware} routing across all possible traffic demands.

\begin{theorem}\label{thm: unsplit}
	Oblivious routing on any undirected capacitated and connected graph $G$ is achievable with performance ratio at most $O\left(d_{max} \cdot \log^2 n + \log n \cdot \log_{\bar{\lambda}} \frac{\pi_{min}}{2} \right)$ with probability $1-\frac{1}{n}$.
\end{theorem}


We note that, somewhat surprisingly, for expander graphs, the resulting $O(\log^2 n)$ upper bound matches the performance of the state-of-the-art \emph{demands-aware} algorithms~\cite{chakrabarti2007approximation}. En route to establishing the above result for u-MCF, we generalize the classical Valiant Load Balancing (VLB) scheme, previously applied to hypercubes, and other highly-structured graphs, to general graphs.

\noindent{\bf Results for Valiant's packet-switching model.} Valiant and Brebner~\cite{valiant1981universal,valiant1982scheme} considered routing on the hypercube network under permutation traffic demands, i.e., when each vertex wishes to send a single packet to a single, distinct, other vertex. Under this routing model, only a single packet can traverse each link simultaneously and the goal is to minimize delay, i.e., the total time for all sent packets to arrive at their destinations. The main result in~\cite{valiant1981universal} can be phrased as follows:

\begin{theorem}
    Oblivious routing of permutation traffic demands on the hypercube is achievable with at most $O(\log n)$ delay with probability $1-\frac{1}{n}$.
\end{theorem}

VLB is a classical, and widely applied, routing scheme~\cite{rabin1989efficient,mckeown2009software,xu1994iterative,tavakoli2009applying,greenberg2008towards,greenberg2009vl2,zhang2005designing}. Adaptations of VLB have been devised for \emph{specific}, well-structured graphs (see, e.g.,~\cite{upfal1992log,aleliunas1982randomized}). Importantly, all these schemes heavily utilize the structure of the graph in selecting paths (e.g., bit-fixing in the hypercube) and, consequently, how VLB can be generalized to \emph{arbitrary} graphs is not obvious. Our main result for this context generalizes VLB to general \emph{regular} graphs.

\begin{theorem} \label{thm: VLB on regular graphs}
	Oblivious routing of permutation traffic demands on any connected $d$-regular graph $G$ is achievable in at most $O\left(\log_{\bar{\lambda}} \frac{1}{2n}\cdot \log n + \frac{1}{d} \log^2_{\bar{\lambda}} \frac{1}{2n}\right)$ time with probability $1-\frac{1}{n}$.
\end{theorem}

Theorem~\ref{thm: VLB on regular graphs} implies time bounds of $O(\log^2 n)$ for expander graphs and $O(\log^3 n)$ for the hypercube. The gap between the latter expression and the original guarantee of $O(\log n)$ in~\cite{valiant1981universal} is due to the relatively long mixing time in the hypercube.

\noindent{\bf Results for splittable multicommodity flow.} The study of oblivious routing under the s-MCF model was initiated by R{\"a}cke~\cite{racke2002minimizing}, who later presented an algorithm with performance ratio $O(\log n)$~\cite{racke2008optimal} for this setting, which is asymptotically tight for general graphs~\cite{bartal1997line,maggs1997exploiting,hajiaghayi2006new}.

Our main result for this setting is the following:

\begin{theorem} \label{thm: split_intro}
	Oblivious routing on any capacitated, undirected and connected graph $G$ is achievable with performance ratio at most $12 \cdot \log_{\bar{\lambda}} \frac{\pi_{min}}{2}$.
\end{theorem}

In terms of performance, our algorithm matches the state-of-the-art for expander graphs, yielding an $O(\log n)$ approximation. Importantly, as the performance of the algorithm improves with the spectral gap, novel \textit{constant bounds} are established w.r.t. dense graphs (i.e., of size $n=d^{\alpha}$ for some constant $\alpha\geq 1$) with substantially large spectral gap, which have received much attention in the computer networking and parallel computing realms~\cite{besta2014slim,kim2008technology} . E.g., since a random regular graph satisfies $\lambda = \Theta(d^{-0.5})$ w.h.p.~\cite{friedman2003proof}, applying Theorem \ref{thm: split_intro} to random dense graphs yields, w.h.p., a performance ratio of $\Theta(\log_{\bar{\lambda}} \frac{\pi_{min}}{2}) = \Theta(\log_{\sqrt{d}} d^\alpha) = \Theta(2\alpha)$. Concrete (deterministic) constructions of such graphs include complete graphs and polarity graphs~\cite{brown1966graphs,erdos1962problem} with resulting performance ratios of $12$ and $48$ respectively. 

In terms of running time, computing the routing solution involves $O(n^4 \log_{\bar{\lambda}} \frac{\pi_{min}}{2})$ operations, e.g., $O(n^4 \log n)$ for expanders. We point out that this bound is close to the optimal running time as the output representation for this problem is of size $n^2\cdot |E|$. Consequently, our random-walk-based algorithm improves the running time over the state-of-the-art which is either LP-based~\cite{azar2003optimal} or relies on hierarchical graph decompositions~\cite{racke2002minimizing,racke2005distributed,racke2008optimal}.


\noindent{\bf Organization.} We begin with Section~\ref{sec: split}, where we introduce our results for the s-MCF model. Section~\ref{sec: permutation model} presents our results for Valiant's model, which are then leveraged in Section \ref{sec: unsplittable flow} to devise a routing scheme for the u-MCF model. We conclude with intriguing open questions in Section~\ref{sec: open questions}.


\section{Oblivious Routing of Splittable Multicommodity Flow}\label{sec: split}


\noindent\textbf{The model.} A \emph{demand matrix} is a non-negative matrix $D$, where its $(i,j)$'th entry, $D_{ij}$, specifies the amount of flow that vertex $i$ wishes to send to vertex $j$. A \textit{splittable multicommodity flow} $f=(f_{ij})_{i,j\subset V\times V}$ is a collection of functions $f_{ij}:E\rightarrow\mathbb{R}$ such that for every two vertices $i,j\in V$, the corresponding function $f_{ij}$ is a flow from $i$ to $j$. Namely, it specifies the traffic from $i$ to $j$ that traverses each edge $e\in E$, and must satisfy the standard flow conservation constraints\footnote{(1) $f(x,y) = -f(y,x)$, (2) $\sum_{i\sim y} f_{ij}(i,y) = \sum_{y\sim j} f_{ij}(y,j)$ and (3) $\sum_{x\sim y} f_{ij}(x,y) = 0$ for $x\neq i,j$.}. A \emph{routing policy} for $G$ is a multicommodity flow $r=(r_{ij})_{i\neq j\in V}$ such that each $r_{ij}$ is a unit flow\footnote{i.e.,
$\sum_{i\sim y} r_{ij}(i,y) = \sum_{y\sim j} r_{ij}(y,j) = 1$.}. A routing policy $r$ and demand matrix $D$ induce a flow $f^*=(f^*_{ij})_{i\neq j\in V}$ such that $f^*_{ij}=D_{ij}\times r_{ij}$. Observe that under $f^*$, all demands $d_{ij}$ are satisfied, though edge capacities might be exceeded. 

We present the following definitions of edge congestion and global congestion w.r.t. a routing policy $r$ and demand matrix $D$:
$$\text{EDGE-CONG}_{D,r}(e):=\frac{\sum_{i,j} D_{ij}\cdot |r_{ij}(e)|}{c(e)}, \>\>\>\> \text{CONG}(D,r):=\max_{e\in E} \text{EDGE-CONG}_{D,r}(e)$$

The \textit{oblivious ratio} of a routing policy $r$ is $PERF(r):= \sup_D \frac{\text{CONG}(D,r)}{OPT(D)}$, where $OPT(D)$ refers to the optimal congestion across all possible (splittable) flows.

\noindent\textbf{The algorithm.} We present a deterministic oblivious routing scheme for arbitrary demand matrices and capacitated undirected graphs. The scheme is specified in Algorithm \ref{alg: split}, where: (1) $A$ is the random walk matrix of $G$, (2) the \textit{point-wise multiplication} of a vector $v$ and a matrix $M$ is the $n\times n$ matrix 
$(v*M)_{xy} := v_x \cdot M_{xy}$, (3) $M^T$ is the transpose of $M$, and (4) the \textit{row normalization} of $M$ is given by:
$\text{row-norm}(M)_{xy}:= \left\{\begin{matrix} 
\frac{M_{xy}}{\sum_z M_{xz}} & \mbox{if } \sum_z M_{xz} \neq 0  \\ 
\frac{1}{d_x}\textbf{1}_{(x,y)\in E} & \mbox{otherwise }  \end{matrix}\right.
$.

\begin{algorithm}[th]\caption{Oblivious routing scheme for splittable MCF} \label{alg: split}
	\begin{algorithmic}[1]
		\STATE {\bf Input:} An undirected graph $G=(V,E,c)$.
		\STATE {\bf Output:} Oblivious routing policy $r=(r_{ij})_{i\neq j\in V}$.
		
		\STATE{\textbf{set} $G' = G$ + self loops}
		\IF{$\lambda(G')<\lambda(G)$}
            \STATE{\textbf{set} $G\leftarrow G'$} 
        \ENDIF
        \STATE{\textbf{set} $\bar{\lambda} = \lambda(G), A=A(G), \pi = \pi(G)$}
        \STATE{\textbf{set} $k=\log_{\bar{\lambda}} \frac{\pi_{min}}{2}$}
        
		\FOR{\texttt{$i\neq j\in V$}}
		\STATE \textbf{set} $v_{ij}^{(0)} = e_i$
		\FOR{\texttt{$1\leq s \leq k$}}
		\STATE $M_{ij}^{(s)} = A$ 
		\STATE $v_{ij}^{(s)} = e_i A^s$
		\ENDFOR
	    \FOR{\texttt{$1\leq s \leq k$}}
		\STATE $M_{ij}^{(k+s)} =  \text{row-norm}\left[\left(e_j A^{k-s} * A\right)^T\right]$
		\STATE $v_{ij}^{(k+s)} = v_{ij}^{(k+s-1)} M_{ij}^{(k+s)}$
		\ENDFOR
		
		\STATE \textbf{set} $r_{ij} = \sum_{s=1}^{2k} \left[ \left(v_{ij}^{(s-1)}*M_{ij}^{(s)}\right) - \left(v_{ij}^{(s-1)}*M_{ij}^{(s)}\right)^T \right]$
		\ENDFOR
				
	\end{algorithmic}
\end{algorithm}

\noindent\textbf{Proof overview.} Our high level approach to the analysis of Algorithm \ref{alg: split} relies on a \textit{sequential routing scheme} that routes the data from $i$ to $j$ across the graph in discrete time steps, $1\leq s\leq 2k$. The first phase of this scheme, where $1\leq s\leq k$, is starting at vertex $i$ and applying $k$ random walk operations. We shall show that the second phase, where $k+1\leq s\leq 2k$, is dominated by an ``inverse'' of a random walk that begins in $j$. The significance of this correspondence is twofold: First, it implies that $r_{ij}$ is a unit-flow from $i$ to $j$. Furthermore, it allows us to upper bound the congestion induced by $r$ by the congestion induced by the random walk, which we relate to the optimal possible congestion. The performance guarantee follows.

\noindent\textbf{Proof of Theorem~\ref{thm: split_intro}.} We defer some of the proofs of this section to Appendix~\ref{app: proof split}. All vectors are assumed to be row vectors in $\mathbb{R}^n$ (using the standard convention w.r.t. Markov chains). 

Recall that $\bar{\lambda}(G)=1$ iff $G$ is either disconnected or bipartite. As the first case does not hold by the assumption of the theorem, and the latter is eliminated by the addition of self loops in the first few lines of Algorithm \ref{alg: split}, we conclude that $\bar{\lambda}<1$, hence $k=\log_{\bar{\lambda}} \frac{\pi_{min}}{2}<\infty$. The sequential routing scheme that lies at the heart of Algorithm \ref{alg: split} is composed of $2k$ linear operators, $M_{ij}^{(1)}, ..., M_{ij}^{(2k)}$, where $M_{ij}^{(s)}$ specifies the distribution of data from each vertex to its neighbours. The vectors $v_{ij}^{(0)}, ..., v_{ij}^{(2k)}$, in turn, represent the data distribution in each time step. Namely, $v_{ij}^{(s)}(x)$ corresponds to the amount of data sent from $i$ to $j$ that is stored at $x\in V$ after $s$ time steps. While the definition of the first $k$ operators is straightforward, with $M_{ij}^{(s)} := A$ for all $1\leq s\leq k$, the choice of the latter operators requires explanation, and in fact reflects the aim of ``inverting the random walk''. Why not simply use $A^{-1}$? Well, first of all, the random walk matrix need not be invertible (for sufficient conditions on the invertibility of adjacency matrices, see, e.g., \cite{mcleman2014graph}, \cite{sciriha2007characterization}). More importantly, we shall require additional properties from our operators, that $A^{-1}$, even when it does exist, does not necessarily posses: (1) A matrix $M\in M_n(\mathbb{R})$ is said to \textit{respect a graph} $G=(V,E)$ if $M_{xy}\neq 0$ only if $(x,y)\in E$. (2) A matrix is said to be \textit{right stochastic} if each of its entries is non-negative with each row summing to 1.

The following Lemma justifies the choice of $M_{ij}^{(k+s)} =  \text{row-norm}\left[\left(e_j A^{k-s} * A\right)^T\right]$ :

\begin{lemma} \label{lem: rev_matrix}
Let $A\in M_n(\mathbb{R})$ be a random walk matrix of an undirected capacitated graph $G$, then, for any non-negative vector $v\in \mathbb{R}^n$, the operator $M:=\text{row-norm}\left[\left(v*A\right)^T\right]$ is right stochastic, respects $G$ and satisfies $vAM =v$.
\end{lemma}

Applying Lemma \ref{lem: rev_matrix} for $v = e_j A^{k-s+1}$, we have

\begin{equation}\label{eq:M = rev(v,A)}
e_j A^{k-s+1} M_{ij}^{(k+s)} = e_j A^{k-s}
\end{equation}

This fact allows us to bound the distributions $v_{ij}^{(k+s)}$:

\begin{lemma} \label{lem: v < 3e_j}
	The vectors $v_{ij}^{(k+1)}, v_{ij}^{(k+2)},..., v_{ij}^{(2k)}$ obtained by Algorithm \ref{alg: split}, \newline satisfy $v_{ij}^{(k+s)}\leq 3e_j A^{k-s}$ for all $0\leq s \leq k$.
\end{lemma}
 
Where we use the notation $v\leq v'$ to indicate that $v_x \leq v'_x$ in every entry $1\leq x\leq n$.
\begin{proof}(of Lemma \ref{lem: v < 3e_j})
	The choice of  $k = \log_{\bar{\lambda}} \frac{\pi_{min}}{2}$ implies that every probability vector $v$ satisfies
	
	\begin{equation}\label{eq: mixing bounds v}
	\left\| v A^k - \pi \right\|\leq \bar{\lambda}^k = \frac{\pi_{min}}{2}.
	\end{equation}
		
	In particular, every entry of the probability vectors $e_i A^k, e_j A^k$ satisfies
	
	\begin{equation}\label{eq: mixing bounds e_i}
	\frac{\pi_x}{2}\leq \pi_x - \frac{\pi_{min}}{2}\leq (e_i A^k)_x, (e_j A^k)_x \leq \pi_x + \frac{\pi_{min}}{2} \leq \frac{3\pi_x}{2}
	\end{equation}
		
	This immediately implies that $v^{(k)}_{ij} = e_i A^k\leq 3e_j A^k$. We continue inductively: Assuming the existence of a non-negative vector $w$, with $v_{ij}^{(k+s-1)} + w = 3e_j A^{k-s+1}$, we have
	
	\begin{equation*}
	\begin{split}
	v_{ij}^{(k+s)} +w M_{ij}^{(k+s)} 
	& = v_{ij}^{(k+s-1)}M_{ij}^{(k+s)} +w M_{ij}^{(k+s)} \\
	& = \left[v_{ij}^{(k+s-1)} +w\right] M_{ij}^{(k+s)} = 3e_j A^{k-s+1} M_{ij}^{(k+s)} = 3e_j A^{k-s}
	\end{split}
	\end{equation*}
	
	where the first equality stems from the definition of $v_{ij}^{(k+s)}$, and the last from equation \ref{eq:M = rev(v,A)}. Since both $w$ and $M_{ij}^{(k+s)}$ consist of non-negative entries, we conclude that $v_{ij}^{(k+s)}\leq 3e_j A^{k-s}$.
\end{proof}

By Lemma \ref{lem: rev_matrix}, the operators $M_{ij}^{(1)}, M_{ij}^{(2)},..., M_{ij}^{(2k)}$ are right stochastic, hence the sum of entries of $v_{ij}^{(2k)} := e_i M_{ij}^{(1)} M_{ij}^{(2)},... M_{ij}^{(2k)}$ must equal $1$. Since Lemma \ref{lem: v < 3e_j} implies that $v_{ij}^{(2k)} \leq 3e_j$, it follows that $v_{ij}^{(2k)} = e_j$. We now show that the resulting function $r_{ij}$ is indeed a unit-flow from $i$ to $j$.

\begin{lemma}\label{lem: r_ij is a flow}
	The matrix $r_{ij}$ determines a unit-flow on $G$ from $i$ to $j$. 
\end{lemma}

In order to analyze the congestion incurred by $r=\left(r_{ij}\right)_{i,j\in V}$ w.r.t. $D$, we first define the \textit{sequential traffic}, and the \textit{sequential congestion} at time step $s$ by

\begin{equation*}
TRAF_D^{(s)} := \sum_{ij} \left(D_{ij}\cdot v_{ij}^{(s-1)}\right)*M_{ij}^{(s)};\;\;\;\; CONG_D^{(s)}(x,y) := \frac{TRAF_{D}^{(s)}(x,y)}{c(x,y)}
 \end{equation*}

Namely, $TRAF_D^{(s)}, CONG_D^{(s)}\in M_n(\mathbb{R})$ where their $(x,y)$'th entry correspond to the traffic/congestion incurred by the sequential routing scheme at time step $s$ on the edge $(x,y)\in E$. We shall compare these matrices with the sequential traffic and congestion obtained by repeated iterations of the random walk operator $A$ over some initial distribution $v\in\mathbb{R}^n$. We define these parameters by $\textit{RW-}TRAF_v^{(s)} := vA^{s-1}*A;\;\; \textit{RW-}CONG_v^{(s)}(x,y) := \frac{\textit{RW-}TRAF_v^{(s)}(x,y)}{c(x,y)}$.

The following lemma establishes the relation between the congestion in both processes, and an upper bound in terms of the demand matrix $D$:

\begin{lemma}\label{lem: seq vs. RW cong}
We have the following:
    $$CONG_D^{(s)}(x,y) \leq \max_{(x,y)\in E, \; 1\leq s\leq k} \textit{RW-}CONG_{\textbf{1}_n D^T}^{(s)} (x,y) \leq \max_x \left\{\frac{\sum_z D_{xz}}{d_x}\right\}_{x\in V}$$

    $$CONG_D^{(k+s)}(x,y) \leq 3\cdot \max_{(x,y)\in E, \; 1\leq s\leq k} \textit{RW-}CONG_{\textbf{1}_n D}^{(k+s)} (y,x) \leq 3\cdot \max_x \left\{\frac{\sum_z D_{zx}}{d_x}\right\}_{x\in V}$$
\end{lemma}

We now relate these quantities to the optimal (total, non-sequential) congestion: Given a demand matrix $D$, any flow on $G$ is required to deliver $\sum_z D_{xz}$ amount of data from $x$ to its neighbors. Since an even distribution of the congestion between the edges adjacent to $x$ yields a congestion of $\frac{\sum_z D_{xz}}{d_x}$ per edge, it follows that this amount of congestion is inevitable under any routing scheme. Similarly, any proper routing should deliver $\sum_z D_{zx}$ to $x$ from its neighbors. It follows that at least $\frac{\sum_z D_{zx}}{d_x}$ congestion is incurred. We conclude that $\max_x \left\{\frac{\sum_x D_{xz}}{d_x}, \frac{\sum_x D_{zx}}{d_x}\right\} \leq OPT(D)$.

Applying the above inequality and Lemma \ref{lem: seq vs. RW cong}, we asserr that $CONG_D^{(s)}(x,y) \leq 3 OPT(D)$ for every $1\leq s\leq 2k$. We are now able to complete the proof using the fact that $r$ is induced by the sequential routing scheme:
\begin{lemma}\label{lem: cong < OPT}
    $CONG_{D,r} (x,y)\leq 12\left(\log_{\bar{\lambda}} \frac{\pi_{min}}{2}\right)OPT(D)$
\end{lemma}

The analysis is tight, as it is easy to construct a demand matrix for which the output of Algorithm \ref{alg: split} satisfies $PERF(r) = \Omega\left(\log_{\bar{\lambda}} \frac{\pi_{min}}{2}\right)$ (See Appendix \ref{app: proof split} for elaboration).




\section{Valiant Load Balancing on Arbitrary Graphs}\label{sec: permutation model}

\begin{algorithm}[th]\caption{VLB on Arbitrary Graphs} \label{alg: VLB}
	\begin{algorithmic}[1]
		\STATE {\bf Input:} An undirected capacitated graph $G=(V,E,c)$.
		\STATE {\bf Output:} A set of paths $\{\gamma_{xy}\}_{x\neq y\in V}$, where $\gamma_{xy}$ is a path from $x$ to $y$.
		
		\STATE{\textbf{set} $G' = G$ + self loops}
		\IF{$\lambda(G')<\lambda(G)$}
            \STATE{\textbf{set} $G\leftarrow G'$} 
        \ENDIF
        \STATE{\textbf{set} $\bar{\lambda} = \lambda(G), A=A(G), \pi = \pi(G)$}
        \STATE{\textbf{set} $k=\log_{\bar{\lambda}} \frac{\pi_{min}}{2}$, $m=\frac{24\log n}{\pi^2_{min}}$}
        
        \STATE{\textbf{set}} $B_{xy}=\emptyset$ for all $x,y\in V$
        		
        \FOR{all $x\in V$}
        \STATE start $m\pi_x$ random walks of length $k$ from $x$
        \IF{the resulting path terminates in $y\in V$}
        \STATE{store the path in $B_{xy}$}
        \ENDIF
        \ENDFOR
        
		\FOR{\texttt{$x\neq y\in V$}}
		\STATE choose a random vertex $r(x,y)$ w.r.t. the stationary distribution: $Pr[r(x,y)=z]=\pi_z$
		\STATE choose $\alpha\in B_{x,r(x,y)}$ uniformly at random.
		\STATE choose $\beta\in B_{r(x,y),y}$ uniformly at random.
		\STATE{set $\gamma_{xy}=\alpha*\beta$}
		\ENDFOR
	\end{algorithmic}
\end{algorithm}

\noindent\textbf{The model.} Under Valiant's classical model~\cite{valiant1981universal,valiant1982scheme} for the routing of packets in networks, each vertex in $G$, aims to send a single (unsplittable) \emph{packet} to a single other vertex in $G$. Specifically, the communicating pairs are determined by a permutation $\sigma:[n]\rightarrow[n]$ such that vertex $i\in[n]$ wishes to send a packet to vertex $\sigma(i)$. Routing on this graph is a discrete-time process; the transmission of a packet across an edge takes a single time step, and packets can traverse an edge only one at a time. 

A \emph{routing policy} $r$ in this model is a set of paths $(\gamma_{ij})_{i\neq j\in V}$. A routing policy $r$ is \emph{oblivious} if $r$ does not depend on the permutation demands $\sigma$. Observe that every routing policy $r$ and permutation demands $\sigma$ induce a flow of packets in the network in which each packet from $i$ to $j$ traverses the (single) path $\gamma_{ij}$ and when more than a single packet needs to traverse an edge $e$, packets are sent across the edge consecutively (say, according to some lexicographic ordering over the packets). We can now define $\text{DELAY}(r)$, for an oblivious routing policy $r$, to be the worst-case delay across all permutation demands $\sigma$.

\noindent\textbf{The algorithm.} The main obstacle facing the generalization of VLB to arbitrary graphs is the absence of a clear definition of ''canonical paths''. We begin, as specified in Algorithm~\ref{alg: VLB}, by generating a ``sample space'' of paths in the graph by randomly sampling multiple fixed-length random walks from each vertex. We then route traffic between every two vertices by selecting an intermediate vertex at random and then concatenating a randomly selected path from the sample space connecting the source to that intermediate vertex and a randomly selected path from the sample space connecting the intermediate vertex to the destination.

While the guarantees of Theorem \ref{thm: VLB on regular graphs} established in this section are for \emph{regular} graphs, the algorithm is applicable to arbitrary \textit{irregular} and \textit{capacitated} graphs. We shall leverage this fact later (Section \ref{sec: unsplittable flow}), when we apply this algorithm to the u-MCF context.

\noindent\textbf{Proof overview.} We first assert that w.h.p. the sample space consists of paths between every pair of vertices. The main insight is that the resulting paths are well-distributed across the graph in the following sense: For each edge in the graph, the expected number of vertices whose paths traverse that path is approximately the same. This allows us to bound the expectation of the number of paths traversing each edge and probability that there exists an edge that is traversed by ``too many'' paths. An upper bound on the delay of each packet follows.

While this section is focused on \emph{regular} graphs, most of the results proven below apply to arbitrary \textit{irregular}, \textit{non-capacitated} graphs. We will utilize the machinery developed in this section to prove Theorem~\ref{thm: unsplit} in Section \ref{sec: unsplittable flow}.

\noindent\textbf{Proof of Theorem \ref{thm: VLB on regular graphs}.} Some of the proofs in this section are deferred to Appendix \ref{app: proof VLB}.

Using the addition of self-loops, if needed, we first assert that $\bar{lambda}<1$ (as discussed in Section \ref{sec: split}). We now construct a sample space of paths as follows: For every vertex $x\in V$, start $m\pi_x$ independent random walks starting at that vertex, each of length $k$. Let $\Omega$ denote the set of all resulting paths, and let $B_{x,y}$ denote the set of paths in $\Omega$ with end-points at $x$ and $y$. In the case of regular graphs, we have $k=\log_{\bar{\lambda}} \frac{1}{2n}$ and $|\Omega|=m=24n^2\log n$. Applying the Chernoff bounds, we show the following:

\begin{lemma}\label{lem: B_xy not empty}
$B_{x,y}\neq \emptyset$ for all $x,y\in V$ with probability $1-\frac{1}{n}$.
\end{lemma}

Now, for every permutation over the vertices, the routing scheme provides us with $2n$ paths: $\alpha_1,...,\alpha_n, \beta_1,...,\beta_n$. We shall show that the first $n$ of those are well-distributed across the graph, and the same result to the rest is due to symmetry. Let $e\in E$ be an edge, and let $W_e$ be defined by $W_e:=\frac{1}{\pi_{max}}\sum_x \pi_x \cdot \textbf{1}_{e\in \alpha_x}$.

This random variable can be interpreted as the total mass of data that traverses $e$ when each vertex $x\in V$ sends a packet of size $\frac{\pi_x}{\pi_{max}}$ through $\alpha_x$. This is equivalent to identifying each path with a weight, proportional to its origin's degree. In the case of a regular graph, each vertex uniformly transmits a packet of size $1$, and $W_e$ corresponds to the number of paths that traverse through $e$ in the routing scheme. We now assert that $E[W_e]\sim E[W_{e'}]$ for any two edges $e,e'\in E$. 

\begin{lemma} \label{lem: similar expectations}
	Let $e,e'\in E$ be edges in the graph, then $E[W_e]\leq 3\cdot E[W_{e'}]$.
\end{lemma}

Suppose that each vertex $x\in V$ sends $\frac{\pi_x}{\pi_{max}}$ amount of data to $\sigma(x)$ via the proposed routing scheme. Summing the flow induced on all edges, we have

\begin{equation}\label{eq: sum of E[W_e]}
\sum_{e} E\left[W_e\right] = E\left[\sum_{e} W_e\right] = E\left[\frac{1}{\pi_{max}}\sum_{x} \pi_x\cdot |\alpha_{x,r(x)}|\> \right] = \frac{k}{\pi_{max}} 
\end{equation}

Equation \ref{eq: sum of E[W_e]} together with Lemma \ref{lem: similar expectations} imply (See Appendix \ref{app: proof VLB} for elaboration) that

\begin{equation} \label{eq: est W_e}
\frac{2}{3}\cdot \frac{k}{d_{max}} \leq E[W_e] = \mu_e \leq 6\cdot \frac{k}{d_{max}}
\end{equation}



Since $W_e$ is the sum of $n$ independent random variables, we apply the Chernoff bounds to bound the probability that $W_e$ is large. Applying the union bound over all edges simultaneously yields:

\begin{lemma}\label{lem: Pr[exist e] < 1/n^r}
For every $r>0$, we have $\Pr \left(\exists e\in E \> s.t. \>  W_e > 9(2+r)\log n + \frac{18k}{d_{max}} \log_{\bar{\lambda}} \frac{\pi_{min}}{2} \right)\leq \frac{1}{n^r}$.
\end{lemma}

Applying Lemma \ref{lem: Pr[exist e] < 1/n^r} for regular graphs with $r=1$ imply that, w.h.p., all edges $e\in E$ satisfy $W_e = O(\log n + \frac{1}{d}\log_\lambda \frac{1}{n})$. In particular, since $W_e$ corresponds to the number of paths from $\alpha_1,...,\alpha_n$ that coincide with $e$ (for regular graphs), and the delay of a packet is upper bounded by length of the path times the number of coincidences with other paths, every vertex $x\in V$ satisfies  

    \begin{equation}
	DELAY(x) = O\left(\log n \cdot \log_{\bar{\lambda}} \frac{1}{n} +  \frac{1}{d}\log^2_{\bar{\lambda}}  \frac{1}{n}\right)
    \end{equation}

with probability $1-\frac{1}{n}$. Theorem \ref{thm: VLB on regular graphs} follows. While one can easily construct an example with $\Theta(\log n)$ (See Appendix \ref{app: proof VLB}), it is not clear whether a $\Theta(\log^2 n)$ exists in this scenario. In this sense, the tightness of Algorithm \ref{alg: split} remains an open question.

\section{Oblivious Routing of Unsplittable Flows} \label{sec: unsplittable flow}

\noindent{\bf The model.} Recall the definitions for splittable multicommodity flow in Section~\ref{sec: split}. Under the u-MCF model, a flow from each vertex $i$ to another vertex $j$ can only traverse a single path between the two. The definitions of edge congestion, global congestion, and oblivious ratio are analogous to the definitions presented above for splittable flow, only that now the constraint of routing along a single path is enforced. We stress that our bounds on oblivious ratio in this model are actually with respect to the same $OPT(D)$ as in the s-MCF model, i.e., the optimal congestion across all possible \emph{splittable} multicommodity flows (and not merely over \emph{unsplittable} flows).

\noindent{\bf The algorithm.} Our algorithm for u-MCF simply applies Algorithm \ref{alg: VLB}, presented in the context of Valiant's model (Section~\ref{sec: permutation model}), to this context and routes commodities along the computed paths. Recall that while our results for Valiant's model pertain to regular, non-capacitated graphs, Algorithm \ref{alg: VLB}, as stated, is applicable to irregular, capacitated, graphs.

\noindent{\bf Proof overview.} In contrast to our results for Valiant's model, our results for the u-MCF model require establishing the performance guarantees of Algorithm~\ref{alg: VLB} w.r.t. \textit{capacitated} and \textit{irregular} graphs, and under \textit{arbitrary demands}. To this end, our proof of Theorem~\ref{thm: unsplit} involves showing reductions from this general setting to the results presented in Section \ref{sec: permutation model} for Valiant's model, i.e., for \textit{non-capacitated} graphs and under \textit{''canonical demands''}, where each vertex attempts to send one unsplittable commodity of volume proportional to its degree. To establish our upper bound on congestion, we translate the demand matrix into a linear sum of ''canonical demand'' matrices. We apply the machinery introduced in the proof of Theorem \ref{thm: VLB on regular graphs} to each of these matrices separately and then apply the union bound to conclude the proof.

\noindent{\bf Proof of Theorem \ref{thm: unsplit}.} Given a capacitated graph $G$, consider the uncapacitated graph $G'$ obtained from $G$ by decomposing each edge $e\in E$ to $c(e)$ edges of capacity $1$ (assume integer capacities). Using the correspondence between the output of Algorithm \ref{alg: VLB} on both graphs we first assert that it is enough, wlog, to show that the Theorem holds w.r.t. uncapacitated graphs (See Appendix \ref{app: unsplit} for elaboration). 

As discussed in Section \ref{sec: split}, we always have $\max_i \left\{\frac{\sum_j D_{ij}}{d_i}, \frac{\sum_j D_{ji}}{d_i}\right\} \leq OPT(D)$. Assume, wlog, that $\max_i \left\{\frac{\sum_j D_{ij}}{d_i}, \frac{\sum_j D_{ji}}{d_i}\right\} = \frac{\sum_j D_{xj}}{d_x}$ for some $x\in V$, and consider the \textit{row-normalized demand matrix} $\widetilde{D}$ with $\widetilde{D}_{ij} = \frac{d_{max}}{d_i}\cdot D_{ij}$. Let $M=\max_{i,j\in V} \widetilde{D}_{ij}$ be the maximal entry of $\widetilde{D}$. Since $D$ is non-negative, so is $\widetilde{D}$, and therefore the maximal sum of rows in $\widetilde{D}$ satisfies $M\leq \max_i \sum_j \widetilde{D}_{ij} \leq Mn$. This quantity can thus be expressed as $Ms$ for some $1\leq s\leq n$.  We conclude that

\begin{equation}\label{eq: OPT unsplit}
OPT(D)\geq \max_i \left\{\frac{\sum_j D_{ij}}{d_i}\right\} 
= \frac{1}{d_{max}} \cdot \max_i \left\{\sum_j \widetilde{D}_{ij}\right\}
= \frac{Ms}{d_{max}}
\end{equation}

We now \textbf{rearrange the indices.} Let $\gamma$ be a path of length $2k$. We shall show that with high probability, the congestion incurred by $\{\alpha_{xy}\}_{x,y\in V}$ on $\gamma$ is at most $Ms\cdot O(\log^2 n)$. 
For each vertex $i\in V$ arrange the entries of the $i$'th row of $D$, $\{D_{ij}\}_{j\in V}$ by their order of magnitude: $D_i^{(n)}\leq ...\leq D_i^{(2)}\leq D_i^{(1)}$. This ordering induces an order on the paths $\{\alpha_{xy}\}_{x,y\in V}$: The path $\alpha_{xy}$ used to deliver a message of size $D_{xy}$ corresponds to $ D_x^{(t)}$ for some $1\leq t\leq n$. We can thus denote $\alpha_{xy}$ by $\alpha_x^{(t)}$. Now, fix $1\leq t\leq n$ and note that the routing scheme determines a set of paths $\{\alpha_x^{(t)}\}_{x\in V}$, chosen by the same random procedure as $\{\alpha_x\}_{x\in V}$ depicted in the proof of Theorem \ref{thm: VLB on regular graphs}. Let $e\in E$ be an edge, and let $Flow_e^{(t)}$ denote the flow incurred on $e$ by the routing of $\{D^{(t)}_x\}_{x\in V}$ through the paths $\{\alpha_x^{(t)}\}_{x\in V}$. We now have
	
$$Flow_e^{(t)} = \sum_x D_x^{(t)}\cdot \textbf{1}_{e\in \alpha_x^{(t)}} = \sum_x \frac{d_x}{d_{max}} \widetilde{D}_x^{(t)}\cdot \textbf{1}_{e\in \alpha_x^{(t)}}= \frac{1}{\pi_{max}} \sum_x \pi_x \cdot \textbf{1}_{e\in \alpha_x^{(t)}} \cdot \widetilde{D}_x^{(t)}$$

The new order of the indices allows us to formulate the congestion over $e$ as follows:

$$CONG_e(D,r) = \sum_{x,y} D_{xy} \cdot \textbf{1}_{e\in \alpha_{x,y}} 
= \sum _t \left(\sum_x D_x^{(t)}\cdot \textbf{1}_{e\in \alpha_x^{(t)}}\right) = \sum _t \left(\frac{1}{\pi_{max}} \sum_x \pi_x \cdot \textbf{1}_{e\in \alpha_x^{(t)}} \cdot \widetilde{D}_x^{(t)}\right)$$

Now, fix $1\leq x\leq n$. Since $M=\max_{i,j} \widetilde{D}_{ij}$, for all $1\leq t\leq n$, we have $\widetilde{D}_x^{(t)}\leq M$. Whenever $s+1 \leq t$, we can say something stronger: The fact that the size of each row is at most $Ms$, combined with the existence of  $t$ entries of size at least $\widetilde{D}_x^{(t)}$, imply that $\widetilde{D}_x^{(t)}\leq \frac{Ms}{t}$. Applying these bounds in the above equation yields

\begin{equation}\label{eq: cong < W_e}
\begin{split}
CONG_e(D,r) & \leq 
\frac{1}{\pi_{max}} \left(\sum _{t=1}^s \sum_x \pi_x \cdot \textbf{1}_{e\in \alpha_x^{(t)}} \cdot M + \sum _{t=s+1}^n \sum_x \pi_x \cdot \textbf{1}_{e\in \alpha_x^{(t)}} \cdot \frac{Ms}{t}\right)\\
& = M \left(\sum _{t=1}^s W_e^{(t)} + s \sum _{t=s+1}^n \frac{W_e ^{(t)}}{t} \right)
\end{split}
\end{equation}

where $W_e^{(t)}:= \frac{1}{\pi_{max}} \sum_x \pi_x \cdot \textbf{1}_{e\in \alpha_x^{(t)}}$ is the same random variable used in Section \ref{sec: permutation model}. Applying Lemma \ref{lem: Pr[exist e] < 1/n^r} with $r=2$ and the union bound implies

\begin{lemma}\label{lem: u-bound W_e}
$\Pr \left(\exists t,\> \exists e\in E \> s.t. \>  W^{(t)}_e > 36\log n + \frac{18}{d_{max}}\cdot \log_{\bar{\lambda}} \frac{\pi_{min}}{2} \right) \leq \frac{1}{n}$
\end{lemma}

Inequality \ref{eq: cong < W_e} and Lemma \ref{lem: u-bound W_e} assert that with probability at least $1-\frac{1}{n}$, we have


\begin{equation*}
\begin{split}
CONG_e(D,r) & \leq Ms \left(36\log n + \frac{18}{d_{max}}\cdot \log_{\bar{\lambda}} \frac{\pi_{min}}{2}\right) \left(1 + 1 \sum _{t=s+1}^n \frac{1}{t} \right)\\
& \leq OPT(D) \cdot d_{max} \cdot O\left(\log n + \frac{1}{d_{max}}\cdot \log_{\bar{\lambda}} \frac{\pi_{min}}{2} \right) \cdot O(\log n)\\
& = OPT(D) \cdot O\left(d_{max} \cdot \log^2 n + \log n \cdot \log_{\bar{\lambda}} \frac{\pi_{min}}{2} \right)
\end{split}
\end{equation*}

\section{Conclusion and Future Research} \label{sec: open questions}

We presented novel oblivious routing algorithms for three extensively studied settings: s-MCF, u-MCF, and Valiant's model. We leave the reader with many open questions, including: (1) How close is the performance of our algorithms to the optimal oblivious routing ratio (e.g., with respect to the spectral gap of the graph)? (2) Can our approach be extended to directed graphs? (3) Can some of the random choices involved in our scheme for u-MCF (and Valiant's model be derandomized (e.g., the process of determining the set of canonical paths)? (4) Can our approach be leveraged to obtain better performance guarantees by exploiting the structure of specific classes of graphs (e.g., LPS~\cite{lubotzky1988ramanujan} and, more generally, Cayley graphs~\cite{lubotzky1995cayley}, and also expanders generated via random permutations~\cite{broder1987second})? (5) Can our approach be utilized to design \emph{distributed} oblivious routing schemes?

\bibliographystyle{plain}
\bibliography{VLB}

\section*{Appendix}
\appendix

\section{Preliminaries}

\subsection{The random walk operator and expander graphs}\label{app: RW prelim}

The \textit{random walk matrix} $A$ of $G$ is given by $$A_{xy}=\left\{\begin{matrix} 
\frac{c(x,y)}{d_x} & \mbox{if } (x,y)\in E  \\ 
0 & \mbox{otherwise }  \end{matrix} \right . $$

This matrix is diagnolizable with real eigenvalues that lie within the interval $[-1,1]$. Let $-1 \leq \lambda_n \leq ...\leq \lambda_2\leq \lambda_1=1$ be the eigenvalues and denote by $\lambda=\lambda(G):=\max\{\lambda_2,|\lambda_n|\}$ its \textit{generalized second eigenvalue}. We call a $d$-regular graph $G$ of size $n$ with $\lambda(G)\leq \lambda$ for some $\lambda <1$, an \textit{$(n,d,\lambda)$-graph}. A family of graphs $\{G_n\}_{n\in \mathbb{N}}$ is an \emph{expander graph family} if there are some constants $d\in \mathbb{N}$ and $\lambda < 1$ such that for every $n$, $G_n$ is an $(n,d,\lambda)$-graph~\cite{arora2009computational}. 

\subsection{Chernoff bounds}

We use the following version of the Chernoff bounds: Let $X_1,... , X_n$ be independent random variables (not necessarily with the same distribution), with $0\leq X_i\leq 1$ for all $i$. Suppose $X=\sum_i X_i$, and let $\mu := E[X] = \sum_i E[X_i]$. Then, given $\delta>0$, we have:

\noindent\textbf{Upper tail}

\begin{equation} \label{chernoff upper tail}
\Pr \big( X > (1+\delta)\mu\big) \leq \exp \left(-\frac{\delta^2}{2+\delta}\cdot \mu\right)
\end{equation}

\noindent\textbf{Lower tail}

\begin{equation} \label{chernoff low tail}
\Pr \big( X < (1-\delta)\mu\big) \leq \exp \left(-\frac{\delta^2}{2}\cdot \mu\right)
\end{equation}

\section{Proofs for s-MCF}\label{app: proof split}
\begin{proof}(of Lemma \ref{lem: rev_matrix})
	While a straightforward calculation of the entries may verify $vAM =v$, the following proof is more instructive: Using the fact that the sum of the $y$'th row of $\left(v*A\right)^T$ equals $\sum_x v_x A_{xy} = (vA)_y$, we have
	\begin{equation}\label{eq:rev_matrix}
	\left(vA*M\right) = \left(v*A\right)^T  
	\end{equation}
		
	The crux of the proof lies in the interpretation of $(v*A)$ as a representation of the flow induced by the operation of $A$ over $v$. Indeed, the amount of flow that traverse from $x$ to $y$ through the edge $(x,y)$ equals $(v*A)_{xy} = v_x \cdot A_{xy}$. Equation \ref{eq:rev_matrix} now means that whatever amount of flow incurred by the operation of $A$ over $v$ must return via the same edge in the opposite direction when $M$ operates over $vA$, hence $vAM=v$.

	Since the random walk matrix $A$ respects $G$, and this property is preserved through the operations of pointwise multiplication, taking transpose (using the fact that $G$ is undirected) and row-normalization, we assert that $rev(v,A)=\text{row-norm}\left[\left(v*A\right)^T\right]$ respects $G$ as well. Being right stochastic stems directly from the row-normalization operation over a non-negative matrix. 
\end{proof}

\begin{proof}(of Lemma \ref{lem: r_ij is a flow})

	Recall that $r_{ij} := \sum_{s=1}^{2k} \left[ \left(v_{ij}^{(s-1)}*M_{ij}^{(s)}\right) - \left(v_{ij}^{(s-1)}*M_{ij}^{(s)}\right)^T \right]$. Being the sum of anti-symmetric matrices that respect $G$, $r_{ij}$ is anti-symmetric and respects $G$ as well, and thus determines a magnitude and a direction for each $e\in E$. Using the matrix terminology, the standard flow constraints sums up to the equation $\textbf{1}_n r_{ij} = e_j - e_i$. We now show that this is indeed the case:
	
	\begin{equation*}
	\begin{split}
	\textbf{1}_n r_{ij} 
	& = \textbf{1}_n \sum_{s=1}^{2k} \left[ \left(v_{ij}^{(s-1)}*M_{ij}^{(s)}\right) - \left(v_{ij}^{(s-1)}*M_{ij}^{(s)}\right)^T \right]\\
	& = \sum_{s=1}^{2k} \left[ \textbf{1}_n\left(v_{ij}^{(s-1)}*M_{ij}^{(s)}\right) - \left(\left(v_{ij}^{(s-1)}*M_{ij}^{(s)}\right)\textbf{1}_n^T \right)^T \right] \\
	& = \sum_{s=1}^{2k} \left[ v_{ij}^{(s-1)} M_{ij}^{(s)} - v_{ij}^{(s-1)}\right] 
	= \sum_{s=1}^{2k} \left[ v_{ij}^{(s)} - v_{ij}^{(s-1)}\right] \\
	& = v_{ij}^{(2k)} - v_{ij}^{(0)} = e_j - e_i
	\end{split}
	\end{equation*}
		
	Where we used the fact that $\textbf{1}_n\left(v*M\right) = vM$ for every vector $v$ and matrix $M$, and that $\left(v*M\right)\textbf{1}^T_n = v$ in case $M$ is right stochastic.
\end{proof}

\begin{proof}(of Lemma \ref{lem: seq vs. RW cong})
    We begin by asserting that for every $1\leq s\leq k$, we have
    \begin{equation}\label{eq:cong < RW 1}
    CONG_D^{(s)}(x,y) = \textit{RW-}CONG_{\textbf{1}_n D^T}^{(s)} (x,y)
    \end{equation}
and
    \begin{equation}\label{eq:cong < RW 2}
    CONG_D^{(k+s)}(x,y) \leq 3\cdot \textit{RW-}CONG_{\textbf{1}_n D}^{(k-s)} (y,x)
    \end{equation}

Equations \ref{eq:cong < RW 1}, \ref{eq:cong < RW 2} suggest that an upper bound on the sequential congestion of the random walk process suffice in order to bound the sequential congestion in Algorithm \ref{alg: split}.

\noindent\textbf{Bounding the random walk.} Given an initial distribution $v$, let $M$ denote the maximal congestion incurred in $k$ operations of the random walk matrix $v$:
$$M = \max_{(x,y)\in E, \; 1\leq s\leq k} \textit{RW-}CONG_{v}^{(s)} (x,y)$$

A key feature of the random walk operation lies in the fact that the congestion incurred by $x\in V$ is distributed evenly between all edges adjacent to $x$. Indeed, if $A$ operates over $e_x$, the flow induced over the edge $(x,y)$ equals $\frac{c(x,y)}{d_x}$, and thus the congestion incurred on each edge is $\frac{1}{d_x}$. Suppose that $M$ is attained at some edge $(x,y)$ at time step $1<s$. Since $A$ distributes congestion evenly, all the edges $(x,z)$ adjacent to $x$ had suffered the same congestion at time $s$, and the data stored at $x$ at time $(s-1)$ must be of size $M\cdot d_x$. Since the same amount of data had entered $x$ at the previous time step, there must be some neighbor $z$ with

$$\textit{RW-}TRAF_v^{(s-1)}(z,x) \geq \textit{RW-}TRAF_v^{(s)}(x,z) $$
hence

$$\textit{RW-}CONG_v^{(s-1)}(z,x) \geq \textit{RW-}CONG_v^{(s)}(x,z) =M $$ 

We thus assert that the maximal congestion $M$ is incurred right at the first routing step $s=1$. This fact can be seen as a manifestation of the maximum principle (see, e.g., \cite{doyle1984random}), that states that harmonic functions (such as the random walk function) attain their extremal values at the boundary. We conclude that

$$M = \max_{(x,y)\in E} \textit{RW-}CONG_{v}^{(1)} (x,y) = \max_x \left\{\frac{v_x}{d_x}\right\}_{x\in V}$$

Applying this w.r.t. $v = \textbf{1}_n D, \textbf{1}_n D^T$ together with equations \ref{eq:cong < RW 1}, \ref{eq:cong < RW 2} yields the inequalities in the right handsight.

\end{proof}

\begin{proof}(of Lemma \ref{lem: cong < OPT})
Indeed,

\begin{equation*}
\begin{split}
CONG_{D,r} (x,y) & = \frac{\sum_{i,j\in V} D_{ij}\cdot \left|r_{ij}(x,y)\right|}{c(x,y)}\\
& \leq \frac{1}{c(x,y)}\sum_{i,j\in V} \left|\sum_{s=1}^{2k} D_{ij}\cdot\left[ \left(v_{ij}^{(s-1)}*M_{ij}^{(s)}\right) - \left(v_{ij}^{(s-1)}*M_{ij}^{(s)}\right)^T \right](x,y)\right|\\
& \leq \sum_{s=1}^{2k}\left[\frac{\sum_{i,j\in V} D_{ij}\cdot\left(v_{ij}^{(s-1)}*M_{ij}^{(s)}\right)}{c(x,y)} + \frac{\sum_{i,j\in V} D_{ij}\cdot\left(v_{ij}^{(s-1)}*M_{ij}^{(s)}\right)^T}{c(x,y)} \right]\\
& = \sum_{s=1}^{2k}\left[\frac{TRAF_D^{(s)}(x,y)}{c(x,y)} + \frac{TRAF_D^{(s)}(y,x)}{c(x,y)} \right]\\
& = \sum_{s=1}^{2k}\left[CONG_D^{(s)}(x,y) + CONG_D^{(s)}(y,x) \right]\\
& \leq 4k\cdot 3 OPT(D)\\
& = 12\left(\log_{\bar{\lambda}} \frac{\pi_{min}}{2}\right)OPT(D)
\end{split}
\end{equation*}
and the lemma follows.
\end{proof}

\subsection{Tightness}
Algorithm \ref{alg: split} is tight in the following sense: consider an $(n,d,\lambda)$-graph $G$ with the demand matrix $D=Adj(G)$, i.e., every vertex aims to send a message of size $1$ to each of its neighbors. Obviously, $OPT(D)=2$, while our algorithm suffers a congestion of $\frac{2}{d}$ in each step throughout $2k$ steps. It follows that 
$$PERF(r)\geq \frac{4}{d}\cdot \log_{\bar{\lambda}} \frac{1}{2n}$$

\section{Proofs for Valiant's model}\label{app: proof VLB}

\begin{proof}(of Lemma \ref{lem: B_xy not empty})

Since the probability that a walk starting at $x$ terminates at $y$ equals $(A^k e_x)_y$, we have

$$\mu_{x,y}:= E[\>|B_{x,y}|\>]= m\pi_x(e_x A^k)_y + m\pi_y(e_y A^k )_x.$$ 

Under the choice of $k = \log_{\bar{\lambda}} \frac{\pi_{min}}{2}$, equation \ref{eq: mixing bounds e_i} holds, therefore

$$m\pi_x\pi_y \leq \mu_{x,y} \leq 3m\pi_x\pi_y.$$

We now use the lower tail version of Chernoff bound (see \ref{chernoff low tail}) to obtain

$$\Pr \left( |B_{x,y}| < \frac{m\pi_x\pi_y}{2}\right)\leq 
\Pr \left( |B_{x,y}| < \frac{\mu_{x,y}}{2}\right) \leq e^{-\mu_{x,y}/8} \leq e^{-m\pi_x\pi_y/8}.$$

Using the union bound we have that

\begin{equation*}
\begin{split}
\Pr \left(\exists x,y\in V\> \text{s.t.} \> |B_{x,y}| <  \frac{m\pi_x\pi_y}{2}\right)
& = \Pr\left(\bigcup_{x,y} |B_{x,y}| <  \frac{m\pi_x\pi_y}{2}\right)\\
& \leq \sum_{x,y} \Pr \left(|B_{x,y}| <  \frac{m\pi_x\pi_y}{2}\right)\\
& \leq n^2 e^{-m\pi^2_{min}/8}\\
& =e^{2\log n-m\pi^2_{min}/8}\leq \frac{1}{n}
\end{split}
\end{equation*}

where the last inequality holds whenever $\frac{24\log n}{\pi^2_{min}}\leq m$.
\end{proof}

\begin{lemma} \label{lem: e vs. e'}
	Let $z\in V$ be randomly chosen with probability $\pi_z$, and let $\gamma$ denote the path induced by a random walk of finite length, starting at $z$. Then any two edges $e,e'\in E$ satisfy
	$$\Pr(e\in \gamma) =\Pr(e'\in \gamma)$$
\end{lemma}

\begin{proof}(of Lemma \ref{lem: e vs. e'}
	We shall use the fact that the random walk preserves the stationary distribution at each step: In case $|\gamma|=1$, we have
	
	\begin{equation*}
	\begin{split}
	\Pr(\gamma\> \text{ends at}\> y) & = 
	\sum_{x} \Pr(\gamma\> \text{ends at}\> y \mid \gamma\> \text{starts at}\> x)\cdot \Pr(\gamma\> \text{starts at}\> x)\\
	& = \sum_{x\sim y} \frac{1}{d_x}\cdot \pi_x = \frac{d_y}{d_x}\cdot \frac{d_x}{\sum_z dz} = \pi_y 
	\end{split}
	\end{equation*}
	
	Applying induction, we infer that this equality holds for any lenght of $\gamma$. Suppose that $|\gamma|=k$, and let $e_i$ denote the $i$'th edge of $\gamma$. Now,
	
	\begin{equation*}
	\begin{split}
	\Pr\big(e_{i+1}=(x,y)\> \big) 
	& = \Pr\big(e_{i+1}=(x,y)\mid\> \gamma |_i \>\text{ends in} \> x\> \big) + \Pr\big(e_{i+1} = (x,y)\mid\> \gamma|_i \> \text{ends in}\> y\> \big)\\
	& = \frac{1}{d_x}\cdot \pi_x + \frac{1}{d_y}\cdot \pi_y = \frac{2}{\sum_z d_z}
	\end{split}
	\end{equation*}
	
	Where $\gamma|_i$ is the path induced by the first $i$ steps of the random walk. Since this probability is independent of $(x,y)$, we conclude that every edge $e\in E$ coincides $\gamma$ with equal probability:
	$$\Pr(e\in \gamma) =\Pr(e'\in \gamma)$$
	\end{proof}

\begin{proof}(of Lemma \ref{lem: similar expectations})
    
    In order to show that, we shall apply Lemma \ref{lem: e vs. e'} that states that a random walk that begins at a random vertex in the graph, coincides with given edges in the graph with equal probabilities. We now have
	
	\begin{equation} \label{eq:E[N_e]}
	\begin{split}
	E[W_e] & = \frac{1}{\pi_{max}}\sum_{x=1}^{n}\pi_x\cdot\Pr(e\in \alpha_x)\\
	& = \frac{1}{\pi_{max}}\sum_{x=1}^{n} \pi_x \left[\sum_{y=1}^{n} \Pr\big(e\in \alpha_x\mid \> r(x)=y\big)\cdot P\big(r(x)=y\big)\right]\\ 
	& = \frac{1}{\pi_{max}}\sum_{x=1}^{n} \pi_x\sum_{y=1}^{n} \pi_y\Pr\big(e\in \alpha_x\mid \alpha_x \in B_{x,y}\big) 
	\end{split}
	\end{equation}
		
	Fix some $x\in V$ and let $\gamma_x$ be chosen at random from $\bigcup_y B_{x,y}$. Applying Bayes law, we have
	$$\Pr(e\in \gamma_x) = \sum_{y=1}^{n} \Pr(e\in \gamma_x|\gamma_x\in B_{x,y})\cdot \Pr(\gamma_x\in B_{x,y})$$
	
	Applying equation \ref{eq: mixing bounds e_i} yields
	
	$$\frac{\pi_y}{2}\leq \Pr(\gamma_x\in B_{x,y})=(A^k e_x)_y \leq \frac{3\pi_y}{2}$$
		
	We thus conclude that
	\begin{equation} \label{eq: gamma_x}
	\begin{split}
	\frac{1}{2}\sum_{y=1}^{n} \pi_y\cdot \Pr(e\in \gamma_x|\gamma_x\in B_{x,y}) \leq \Pr(e\in \gamma_x) \leq 	\frac{3}{2}\sum_{y=1}^{n} \pi_y\cdot \Pr(e\in \gamma_x|\gamma_x\in B_{x,y})
	\end{split}
	\end{equation}

Observe that both $\gamma_x\in \bigcup_y B_{x,y}$ and $\alpha_x$ from the routing scheme are chosen uniformly when restricted to $B_{x,y}$. In particular, $\Pr(e\in \alpha_x|\alpha_x\in B_{x,y})=\Pr(e\in \gamma_x|\gamma_x\in B_{x,y})$. Multiplying inequality \ref{eq: gamma_x} by $\frac{\pi_x}{\pi_{max}}$ and summing over all $x\in V$, we have:

\begin{equation}
\begin{split}
\frac{1}{2\pi_{max}}\sum_x \pi_x\sum_{y=1}^{n} \pi_y\cdot \Pr(e\in \alpha_x|\alpha_x\in B_{x,y}) 
&\leq \frac{1}{\pi_{max}}\sum_x \pi_x\Pr(e\in \gamma_x)\\
& \leq \frac{3}{2\pi_{max}}\sum_x \pi_x\sum_{y=1}^{n} \pi_y\cdot \Pr(e\in \alpha_x|\alpha_x\in B_{x,y})
\end{split}
\end{equation}

Applying equality \ref{eq:E[N_e]} (note that both sides of this inequality are dominated by $E[W_e]$):
	
$$\frac{1}{2}E[W_e] \leq 
\frac{1}{\pi_{max}}\sum_x \pi_x\Pr(e\in \gamma_x)  
\leq \frac{3}{2}E[W_e].$$
	
We now use the fact that 
$$\sum_x \pi_x\Pr(e\in \gamma_x) = \Pr(e\in \gamma)$$

where $\gamma$ is the path induced by a random walk at legnth $k$ that starts at a random vertex $x\in V$. Applying lemma \ref{lem: e vs. e'}, this probability is equal for every pair of edges $e,e'\in E$. We conclude that

$$\frac{1}{2}E[W_e] \leq 
\frac{1}{\pi_{max}}\sum_x \pi_x\Pr(e\in \gamma_x) 
= \frac{1}{\pi_{max}}\sum_x \pi_x\Pr(e'\in \gamma_x) 
\leq \frac{3}{2}E[W_{e'}] $$

Which yields
	
	$$\frac{E[W_e]}{E[W_{e'}]}\leq 3$$
\end{proof}

\begin{proof}(of Inequality \ref{eq: est W_e})
    
We use the fact that there are $\frac{\sum_{x}d_x}{2}$ edges in the graph, in order to bound the expectation. Applying Lemma \ref{lem: similar expectations}, we infer that

$$\mu_{min}\leq \mu_e:=E[W_e]\leq \mu_{max}$$ 

where $\mu_{max},\>\mu_{min}$ satisfy

$$\mu_{max}+ \left(\frac{\sum_{x}d_x}{2}-1\right)\frac{\mu_{max}}{3} = \frac{k}{\pi_{max}}$$

$$\mu_{min}+ \left(\frac{\sum_{x}d_x}{2}-1\right)\cdot 3\mu_{min} = \frac{k}{\pi_{max}}$$

hence

$$\mu_{max} \leq \frac{6k}{\pi_{max}\sum_x d_x}$$

$$\frac{2k}{3\pi_{max}\sum_x d_x} \leq \mu_{min}$$

and the inequality follows.
\end{proof}

\begin{proof}(of Lemma \ref{lem: Pr[exist e] < 1/n^r})
    
Applying the Chernoff bound (inequality \ref{chernoff upper tail}), we have

\begin{equation*}
\begin{split}
\Pr \left( W_e > (1+\delta)\frac{6k}{d_{max}}\right)
& \leq \Pr \Big( W_e > (1+\delta)\mu_e\Big)\\
& \leq \exp \left(-\frac{\delta^2}{2+\delta}\cdot \mu_e\right)
\leq \exp \left(-\frac{\delta^2}{2+\delta}\cdot \frac{2k}{3d_{max}}\right)
\end{split}
\end{equation*}

We would now like to bound the probability that any edge $e\in E$ obtains large weight $W_e$ under the routing scheme:

\begin{equation}\label{eq: W_e < exp}
\begin{split}
\Pr \left(\exists e\in E \> s.t. \>  W_e > (1+\delta)\frac{6k}{d_{max}} \right) 
& = \Pr \left(\bigcup_e \Big(W_e > (1+\delta)\frac{6k}{d_{max}}\Big) \right)\\
& \leq \sum_{e\in E} \Pr \Big(W_e > (1+\delta)\frac{6k}{d_{max}}\Big)\\
&\leq n^2\cdot \exp \left(-\frac{\delta^2}{2+\delta}\cdot \frac{2k}{3d_{max}}\right)\\
&= \exp \left(2\log n -\frac{\delta^2}{2+\delta}\cdot \frac{2k}{3d_{max}}\right)
\end{split}
\end{equation}

We would like to upper bound this expression with $\frac{1}{n^r}$ for some $r>0$. In order to do so, it suffices to find large enough $\delta$ such that

$$2\log n -\frac{\delta^2}{2+\delta}\cdot \frac{2k}{3d_{max}} \leq -r\cdot \log n$$

or

\begin{equation} \label{eq: lambda^2}
\frac{(2+r)\cdot \log n \cdot 3d_{max}}{2k}
\leq \frac{\delta^2}{2+\delta}
\end{equation}

Since $\delta - 2 \leq \frac{\delta^2}{2+\delta}$, the choice of $\delta = \frac{(2+r)\cdot \log n \cdot 3d_{max}}{2k} + 2$ satisfies inequality \ref{eq: lambda^2}. Applying this to inequality \ref{eq: W_e < exp} yields

\begin{equation}\label{W_e < log n whp}
\begin{split}
\Pr \left(\exists e\in E \> s.t. \>  W_e > 9(2+r)\log n + \frac{18}{d_{max}}\cdot \log_{\bar{\lambda}} \frac{\pi_{min}}{2} \right)\\
=\Pr \left(\exists e\in E \> s.t. \>  W_e > \left(\frac{(2+r)\cdot \log n \cdot 3d_{max}}{2k} + 3\right)\frac{6k}{d_{max}} \right)\\
= \Pr \left(\exists e\in E \> s.t. \>  W_e > (1+\delta)\frac{6k}{d_{max}} \right) \leq \frac{1}{n^r}
\end{split}
\end{equation}
\end{proof}

\subsection{Tightness}
Consider a permutation induced by a perfect matching in a $d$-regular expander graph, where $\sigma(x)$ is the partner of $x\in V$ in the matching. Clearly, two steps suffice to complete the routing. In contrast, routing by the proposed scheme must take at least $2k = \Theta(\log n)$ time steps, as this is the length of the paths used. 

\section{Proofs for u-MCF}\label{app: unsplit}

\textbf{Reduction to uncapacitated graphs \ref{sec: unsplittable flow}.} Suppose that $G$ is equipped with rational capacities, i.e., $c(e)\in \mathbb{Q}$ for all $e\in E$. We first observe that given a flow on the graph, a multiplication of all capacities by some constant factor reduces the congestion induced by the same flow by the same factor. This means that the oblivious ratio is invariant to multiplication of the capacities, and thus it is enough to address only graphs with integer capacities.

	Now, consider  Let $p:G'\rightarrow G$ be the natural projection that corresponds to this decomposition. Fix a path $\gamma$ of length $k$ between a pair of vertices $x,y\in G$. Using the fact that the paths $\alpha_{x,y},\alpha'_{x,y}$ were chosen w.r.t. random walks in $G$ and $G'$, it is easy to verify that

	$$Pr(\alpha_{x,y} = \gamma) = Pr\Big(\alpha'_{x,y} \in p^{-1}(\gamma) \Big)$$ 

	This means that applying the routing policy on $G'$ and then choosing the induced paths on $G$ is equivalent to applying the scheme on $G$ in the first place. Since every flow $f'$ on $G'$ satisfies 
	$$CONG_{G} \big(p(f)\big) \leq CONG_{G'} \big(f'\big)$$

	We conclude that it suffices to analyze the scheme on uncapacitated irregular graphs. Our routing scheme in this scenario is a straightforward extension of the VLB scheme depicted in the proof of Theorem \ref{thm: VLB on regular graphs}. First, we use the same random process in order to obtain a sample space of paths $\Omega=\{B_{x,y}\}_{x,y\in V}$. Now, for every message needed to be sent from $x$ to $y$, choose a random vertex $r(x,y)$ and rout the message from $x$ to $r(x,y)$ and then to $y$ through uniformly random chosen paths $\alpha_{x,y}\in B_{x,r(x,y)}$ and $\beta_{x,y}\in B_{r(x,y),y}$. 
	
\begin{proof} (of Lemma \ref{lem: u-bound W_e})
Fix $1\leq t\leq n$. Since $\{\alpha_x^{(t)}\}_{x\in V}$ are chosen by the random procedure depicted in Algorithm \ref{alg: VLB}, we can apply Lemma \ref{lem: Pr[exist e] < 1/n^r} with $r=2$ and obtain that 
$$\Pr \left(\exists e\in E \> s.t. \>  W^{(t)}_e > 36\log n + \frac{18}{d_{max}}\cdot \log_{\bar{\lambda}} \frac{\pi_{min}}{2} \right) \leq \frac{1}{n^2}.$$

Applying the union bound concludes the proof.
\end{proof}

	
	
	

\end{document}